\documentclass[12pt,oneside,english]{amsart}
\usepackage[T1]{fontenc}
\usepackage[utf8]{inputenc}
\usepackage[letterpaper]{geometry}
\usepackage{float}
\usepackage{booktabs}
\usepackage{amstext}
\usepackage{amsthm}
\usepackage{amssymb}
\usepackage{graphicx}
\usepackage{setspace}
\usepackage[authoryear]{natbib}
\makeatletter

\floatstyle{ruled}
\newfloat{algorithm}{tbp}{loa}
\providecommand{\algorithmname}{Algorithm}
\floatname{algorithm}{\protect\algorithmname}

\theoremstyle{plain}

\include{def}
\def\dispmuskip{\thinmuskip= 3mu plus 0mu minus 2mu \medmuskip=  4mu plus 2mu minus 2mu \thickmuskip=5mu plus 5mu minus 2mu}
\def\textmuskip{\thinmuskip= 0mu                    \medmuskip=  1mu plus 1mu minus 1mu \thickmuskip=2mu plus 3mu minus 1mu}
\def\beq{\dispmuskip\begin{equation}}    \def\eeq{\end{equation}\textmuskip}
\def\beqn{\dispmuskip\begin{displaymath}}\def\eeqn{\end{displaymath}\textmuskip}
\def\bea{\dispmuskip\begin{eqnarray}}    \def\eea{\end{eqnarray}\textmuskip}
\def\bean{\dispmuskip\begin{eqnarray*}}  \def\eean{\end{eqnarray*}\textmuskip}

\usepackage{chngcntr}
\newtheorem{theorem}{Theorem}
\theoremstyle{plain}

\newtheorem{proposition}{Proposition}

\newtheorem{corollary}{Corollary}
\newtheorem{remark}{Remark}
\newtheorem{note}{Note}

\theoremstyle{plain}

\usepackage[within=none]{caption}

\usepackage{subfigure}
\usepackage{babel}
\usepackage{amsmath}
\usepackage{mathtools}
\usepackage{bbm}
\usepackage{relsize}

\usepackage{algpseudocode}

\usepackage [displaymath,mathlines]{lineno}

\def\wh{\widehat}

\def\transp{\tiny T}
\def\bm{\boldsymbol}

\newcommand{\cdummy}{\cdot}
\newcommand{\tmmathbf}[1]{\ensuremath{\boldsymbol{#1}}}

\allowdisplaybreaks

\providecommand{\xequal}[2][]{\mathop{=}\limits_{#1}^{#2}}
\makeatother

\usepackage{babel}
\providecommand{\theoremname}{Theorem}

\newcommand{\mathd}{\mathrm{d}}

\def\transp{{\it{\tiny T}}}
\def\wh{\widehat}
\newcommand{\nobracket}{}

\begin{document}
\title{On approximating copulas by finite mixtures}

\author{Mohamad A. Khaled and Robert Kohn}
\thanks{Khaled: \textit{Department of Economics, University of Queensland, Brisbane Qld 4072, Australia email m.khaled@uq.edu.au}. Kohn: \textit{UNSW School of Business, University of New South Wales, Sydney NSW 2052, Australia emailr.kohn@unsw.edu.au}}
\maketitle

\begin{abstract}
Copulas are now frequently used to construct or estimate multivariate distributions because of their ability to take into account the multivariate dependence of the different variables while separately specifying marginal distributions. Copula based multivariate models can often also be more parsimonious than fitting a flexible multivariate model, such as a mixture of normals model, directly to the data.

However, to be effective, it is imperative that the family of copula models considered is sufficiently flexible. Although finite mixtures of copulas have been used to construct flexible families of copulas, their approximation properties are not well understood and we show that natural candidates such as mixtures of elliptical copulas and mixtures of Archimedean copulas  cannot  approximate a general  copula arbitrarily well.
Our article develops fundamental tools for approximating a general copula arbitrarily well by a copulas based on finite mixtures. We show the asymptotic properties as well as illustrate the advantages of our methodology empirically on a financial data set and on some artificial data.

\end{abstract}

\section{Introduction\label{S: introd}}

The purpose of this paper is to establish a new methodology for the non-parametric estimation of copulas through the use of sieves based on finite mixture models. As good approximation properties of finite mixtures are an indispensable prerequisite for their use in non-parametric estimation, the paper mainly focuses on establishing a mathematical theory for proving approximation properties of those mixtures. Those properties are not only of mathematical interest in their own right, but they also bring the whole machinery of mixture modeling to bear down on the problem of estimating copulas non-parametrically or to a lesser extent modeling them in an optimal way by combining flexibility with parsimony.

Our article provides some foundational tools for using finite mixture models to non-parametrically estimate a target copula function $C$, having density $c$, by asking under what conditions can we find a positive integer $R$, cumulative distribution functions $G_1, \dots, G_R$, and positive probabilities $\pi_1, \dots, \pi_R$ satisfying $\pi_1 + \cdots + \pi_R=1$ to approximate $C$ by $G:=\pi_1 G_1 + \cdots + \pi_R G_R$ for a given precision. Or, alternatively, under what conditions can we find  probability density functions $g_1, \dots, g_R$, to approximate the density $c$ by  $g:=\pi_1 g_1 + \cdots + \pi_R g_R$ for a given precision. Our main result uses this framework to propose a copula family based on a finite mixture that can approximate any copula arbitrarily well. The summary at the end of the paper further discusses the statistical application of these results.

Our article also shows that  neither mixtures of Archimedean copulas nor  mixtures of elliptical copulas can in general approximate a copula density arbitrarily well. Even though they are natural candidate families that have been extensively used to approximate an arbitrary copula, the literature contains no mathematical results on how good their approximation properties are.

The methodology introduced in this paper is not only of interest when it comes to approximating copulas, but it is also useful for approximating arbitrary distributions. Indeed, using a copula-based approach for approximating multivariate distributions involves approximating each of the marginals separately while also approximating the underlying implied joint distribution. Such an approach for approximating multivariate distributions is attractive for two reasons. First, we can directly control the properties of the approximating marginal distributions rather than just deducing their properties from the approximation of the joint. For example, consider approximating a high dimensional multivariate model by a flexible factor based model such as a mixture of factor analyzers; see, for example, Chapter 8 of \cite{McLachlan:Peel:2000}.
It is then difficult to ensure that the implied marginal distributions  will be consistent with an approach that approximates the marginal distributions directly. A second attractive property of copulas is that a copula based multivariate approximation can often be much more parsimonious than approximating the multivariate distribution directly. For example, consider a bivariate distribution with independent marginals each of which is a 6 component mixture of normals. Then approximating this distribution by a bivariate mixture of normals will require a 36 component mixture, while a copula based approach will fit a 6 component mixture to each of the marginals and then a standard normal for the underlying Gaussian copula. Section~\ref{SS: simple illustration} illustrates the same issue on a more complex example. However, it is imperative when using a  copula based approach to approximate multivariate distributions that the family of approximating copulas is sufficiently flexible. The reason is that  the  copula is formed by transforming each of the marginals to a uniform distribution, which can potentially make the underlying distribution of the copula quite complex. \cite{tran2014copula} show empirically  that this can happen, for example, when the original multivariate distribution is heavy tailed or multimodal.

We now briefly review the literature on  nonparametric estimation of copulas. The foundation of non-parametric estimation is based on estimating the copula cumulative distribution function (CDF) using empirical copulas and studying the asymptotic weak convergence properties of the empirical copula process. See, for example, \cite{fermanian2004weak} and \cite{segers2012asymptotics}. Among density estimators, Bernstein copulas constitute a prominent example; see {\cite{sancetta2004bernstein}} and \cite{sancetta2007nonparametric},  or \cite{burda2014copula} for a Bayesian approach. Currently, Bernstein copulas do not scale well with the dimension of the multivariate distribution and applications have been restricted to small dimensions as the number of parameters increases exponentially with dimension. Some other approaches are based on kernels \citep{omelka2009improved} and wavelets \citep{genest2009estimating}. There are very few papers that explicitly address the question of estimating copulas non-parametrically through the use of mixtures. \cite{wu2014bayesian} and \cite{wu2015bayesian} present a Bayesian non-parametric approach. \cite{wu2014bayesian} take Gaussian copulas as the mixture components. \cite{wu2015bayesian} take  multivariate skew normal copulas as the mixture components. However, neither paper presents approximation results.

Our approach has several advantages compared to competing approaches in the literature on dense subfamilies for approximating copulas; 
see, for example, Chapter 4 in \cite{durante2015principles} for a excellent review. First, because our approach is mixture based,
we can use the vast statistical and computational literature on estimating mixtures. Furthermore, compared to some more tractable approximating dense families like Bernstein copulas, our approach is more parsimonious in certain cases because the number of ``effective parameters'' that require 
estimation grows quadratically in the number of dimensions instead of exponentially. Finally, our approach is attractive because it
automatically yields an easily computable valid copula for the approximation.

The rest of the paper is organized as follows. Section~\ref{S: approximation properties} presents our fundamental approximation results and constructs a family of mixture models that can approximate any copula arbitrarily well.
 Section~\ref{SS: simple illustration} uses a simple example to illustrate our approximation approach. Section~\ref{S: asymptotic properties} introduces a concentration inequality useful for illustrating the asymptotic properties of the mixture family.
Sections~\ref{S: mixtures of archimedean copulas} and \ref{S: mixtures of elliptical copulas} respectively
characterize Archimedean and elliptical copulas and discuss their approximation properties. Section~\ref{S: empirical illustrations} applies our  approximation approach to a financial data set that was previously analyzed in the literature. We show that our approach provides a better fit and is more parsimonious than that obtained by a mixture of Gaussian copulas. There are two technical appendices. Appendix~\ref{app: proofs} contains all the proofs. Appendix~\ref{app: drawing from example copula} shows how to sample from the specific copula we use for the illustration in Section~\ref{SS: simple illustration}. Section~\ref{S: conclusion} concludes by discussing future theoretical and computational work based on our results.

\section{Approximation properties of some mixtures of general distributions on
the unit hypercube\label{S: approximation properties}}

We first consider approximating some distributions on the unit interval $(0,1)$, and then consider the case of classes of copulas on $(0, 1)^M$. The reason for starting with the unit interval is that we can introduce our methodology and some main ideas in a simpler setting before focusing on our main objective, which is the approximation of copulas and copula densities.

There is an extensive literature on approximating arbitrary distributions by finite mixtures. See {\cite{zeevi1997density}}, {\cite{dalal1983approximating}} or {\cite{lijoi2003approximating}}. We will use the elements of the theory of approximation by universal series (see
{\cite{bacharoglou2010approximation}} and {\cite{koumandos2010universal}}) to obtain our approximations.

We begin by stating an adaptation of a theorem from {\cite{bacharoglou2010approximation}} that is essential for deriving our main results. The theorem is based on the theory of Universal series in $\bigcap_{p > 1} \ell^p$ and yields approximations in $\| \cdot \|$ over certain subsets of $\mathbb{R}^M$ for bounded continuous functions or functions with bounded support. Here $\ell^p$ stands for the traditional $p$-power summable sequence spaces and $\| \cdot \|$ means either the $L_1$ or the $L_\infty$ norm, or their sum. That norm is used in the statement of Theorem~\ref{bacharoglu_theorem}. Otherwise, if the result concerns a specific norm, then we will denote it by either$\| \cdot \|_1$ or $\| \cdot \|_\infty$ if necessary.
Also, let $\mathbb{N}$ be the set of positive integers and let $\mathbb{Q}$ be the set of rational numbers.

\begin{theorem}
\label{bacharoglu_theorem}
  Let
  \[ \mathcal{A}^+ = \left\{ \alpha = (\alpha_n)_{n \in \mathbb{N}}, \alpha_n
     > 0, \,\,\, \text{for every}~ n, \alpha \in \bigcap_{p > 1} \ell^p \right\},
  \]
  and let the sequence $(\phi_n)_n$ be formed by enumerating
  $\phi_{\frac{1}{k}} (\bm x - \bm \mu)$, where $k \in \mathbb{N}$, $\bm \mu \in
  \mathbb{Q}^M$ (some enumeration of it) and $\phi_s (\bm x - \bm \mu)$ is the
  density of a multivariate normal distribution with mean vector $\bm \mu$ and
  covariance matrix $s\bm I_M$. Let $f$ be some density that has either compact
  support on $\mathbb{R}^M$ or that is bounded and continuous.

  Then, there exists an $\alpha \in \mathcal{A^+}$ and a sequence $(R_n)_n$ of
  integers such that for all $\varepsilon > 0$, there exists an $n (\varepsilon)$
  such that for all integers $n>n(\varepsilon)$, the following holds
  \[ \left\| f - \frac{1}{\sum_{j = 1}^{R_n} \alpha_j}  \sum_{j = 1}^{R_n}
     \alpha_j \phi_j \right\| < \varepsilon .\]
\end{theorem}

\subsection{Approximation on the unit interval\label{SS: approx on an interval}}

Let $F$ be the cumulative distribution function of an absolutely continuous random variable $X$. If one applies the transformation $F$ to $X$, then it immediately follows by a simple calculation that the transformed random variable $F (X)$ is uniform on $(0, 1)$,
\begin{align*}
  \Pr [F (X) \leqslant t] & = \Pr [X \leqslant F^{- 1} (t)]\\
  & =  F \circ F^{- 1} (t)  =  t.
\end{align*}
If one applies a different transformation (say using a different distribution function $H$), then we obtain that the distribution of the transformed random variable $H(X)$ is
\begin{eqnarray*}
  \Pr [H (X) \leqslant t] & = & F \circ H^{- 1} (t).
\end{eqnarray*}
If $H$ is the distribution of an absolutely continuous random variable, then $H(X)$ is an absolutely continuous random variable on $(0, 1)$, but is not uniform in general. Theorem~\ref{thm: one dim approx} exploits such simple constructions by using them as building blocks for finite mixture distributions. In particular, the theorem exploits distributions of the form $F \circ H^{- 1} (t)$ in order to approximate densities (or distributions) on $(0,1)$.

\begin{theorem}\label{thm: one dim approx}
  Let $g$ be an unknown continuous density function with compact support in
  $(0, 1)$. Let $h$ be some arbitrary bounded and absolutely continuous density
  function with its support being the whole real line and let $H$ be its
  corresponding CDF. Let $\phi_{\mu, \sigma}$ be a normal density with mean $\mu$ and standard deviation $\sigma$.
   Then, for every $\varepsilon > 0$, there exist an $R \in \mathbb{N}$, $(\pi_1,
  \ldots, \pi_R) \in \Delta_R$ (the $R$-simplex), $\mu_1, \ldots, \mu_R \in
  \mathbb{R}$ and $\sigma_1, \ldots, \sigma_R \in (0, \infty)$ such that
  \[ \left\| g - \sum_{r = 1}^R \pi_r \frac{\phi_{\mu_r, \sigma_r} \circ H^{-
     1}}{h \circ H^{-1} }\right\| < \varepsilon,   \]
where $\| \cdot \|$ is either the $L_1$ or $L_\infty$ norm.
\end{theorem}

Even though the theorem is stated for the approximation of a density $g$, it can be used with no additional difficulty to approximate an arbitrary continuous cumulative distribution, say $G$. This is a particularly important point, as the results of this paper can be used to approximate either the distribution of a random variable, or, with some extra assumptions, the density of that random variable. The next section extends theorem \ref{thm: one dim approx} to the multivariate setting and shows how to approximate a copula and its density.

\subsection{Approximation of copulas\label{SS: approx of copulas}}

\begin{sloppypar}
Theorem~\ref{corr: copula approx by mixture} applies the same reasoning as the approximation result on unit
intervals to get an approximation result for copulas in terms of distributions on $(0,1)^M$. Let
$H$ be the CDF of an $M \times 1 $ random vector, with marginal distribution functions $H_1, \dots, H_M$ and define
$\mathfrak{F}_H: \mathbb{R}^M \rightarrow (0,1)^M $ as $\mathfrak{F}_H(\bm x):= \left (H_1(x_1), \dots, H_M(x_M) \right )$
We similarly define $\mathfrak{F}_H^{-1}: (0,1)^M \rightarrow \mathbb{R}^M$ as $\mathfrak{F}^{-1}_H(\bm u): = \left (H_1^{-1}(u_1), \dots, H_M^{-1}(u_M) \right )$
where the marginal inverses are appropriately defined.
\end{sloppypar}

\begin{theorem} \label{corr: copula approx by mixture}
  Let $C$ be some arbitrary $M$-dimensional absolutely continuous copula with
  density $c$. Let $H$ be the CDF of an $M \times 1 $ random vector, with marginal CDF's $H_1, \dots, H_M$ that are absolutely continuous with non-compact support on $\mathbb{R} $. Let $h, h_1, \dots, h_M$ be the corresponding densities,
  which we assume are bounded.
 Then, for any $\varepsilon > 0 $, there
 exist $R > 0$, $\pi_1, \ldots, \pi_R
  \in \Delta_R$, $\sigma_1, \ldots, \sigma_r \in (0, \infty)$ and $\bm \mu_1,
  \ldots, \bm \mu_r \in \mathbb{R}^M$ such that
  \[ \left\| c -  q_R(\bm u )\right\|_1 < \varepsilon , \]
     where
\[q_R(\bm u ) :=\sum_{r = 1}^R \pi_r
    \frac{\phi_{\bm \mu_r, \sigma_r I_M} \circ \mathfrak{F}_H^{- 1}}{ \prod_{i=1}^M h_i \circ H_i^{- 1} } (\bm u), \]
     $\phi_{\bm \mu, \sigma \bm I_M  }$ is the density of an $M \times 1 $ normal
     vector with mean $\bm \mu$   and covariance matrix $\sigma \bm I_M$, and $\| \cdot \|_1$ is the $L_1$ norm.
     If $c$ is also continuous, then the result also holds for the $L_\infty$ norm.
\end{theorem}

Notice that $q_R$  is a mixture of distributions on the unit cube whose marginals are
\[ q_{R,i}:=\sum_{r = 1}^R \pi_r \frac{\phi_{\mu_{i, r}, \sigma_r } \circ H^{- 1}_i}{ h_i\circ H_i^{-1} }
    , \quad i=1, \dots, M.  \]
		
Although $q_R$ is not a copula density, it can be shown that it is the density of a random variable whose marginals are arbitrarily close to being uniformly distributed and that, furthermore, the copula of density $q_R$ is arbitrarily close itself to $c$.

\begin{note}

Given the arbitrariness of $H$, the corollary applies even in the simple setting where $H$ is just the CDF of an i.i.d. vector with standard normal marginals.

\end{note}

\begin{note}

Theorem~\ref{corr: copula approx by mixture} provides an approximation result for the density using the $L_1$ norm and can be strengthened to the $L_\infty$ norm under stricter conditions. One must, however, observe that it directly applies to the copula distribution $C$ with no extra assumptions.

That is,  suppose that $C$ is some arbitrary $M$-dimensional absolutely continuous copula distribution function and let $H$ be as in Theorem~\ref{corr: copula approx by mixture}.
 Then, for any $\varepsilon > 0 $, there
 exist $R > 0$, $\pi_1, \ldots, \pi_R
  \in \Delta_R$, $\sigma_1, \ldots, \sigma_r \in (0, \infty)$ and $\bm \mu_1,
  \ldots, \bm \mu_r \in \mathbb{R}^M$ such that
  \[ \left\| C -  q_R(\bm u )\right\| < \varepsilon , \]
     where
\[q_R(\bm u ) :=\sum_{r = 1}^R \pi_r
    \frac{\phi_{\bm \mu_r, \sigma_r I_M} \circ \mathfrak{F}_H^{- 1}}{ \prod_{i=1}^M h_i \circ H_i^{- 1} } (\bm u), \]
     $\phi_{\bm \mu, \sigma \bm I_M  }$ is the density of an $M \times 1 $ normal
     vector with mean $\bm \mu$   and covariance matrix $\sigma \bm I_M$, and $\| \cdot \|$ is either the $L_1$ norm or $L_\infty$ norm.

		Obviously, $R$ and the different probability, location and scale parameters in the approximating mixture are not, in general, the same as in Theorem~\ref{corr: copula approx by mixture}.
		
\end{note}

The next corollary shows that the marginals of $q_R$
can be made arbitrarily close to uniform. Let $Q_R$ be the CDF of $q_R$ with marginals $Q_{R,i}, i=1, \dots, M$.
\begin{corollary} \label{corr: marginals1}
Suppose that the conditions in Theorem~\ref{thm: copula approx by mixture} hold.  Then, for any $\varepsilon > 0 $, there
 exist $R > 0$, $\pi_1, \ldots, \pi_R
  \in \Delta_R$, $\sigma_1, \ldots, \sigma_r \in (0, \infty)$ and $\bm \mu_1,
  \ldots, \bm \mu_r \in \mathbb{R}^M$ such that
 $\| q_{R,i} - \bm 1_{(0,1)}\| < \varepsilon $, for $i=1, \dots, M$ for both the $L_1$ and $L_\infty$ norms.
\end{corollary}
Let $G$ be the cdf of a $M \times 1 $ random vector with density $g$ and marginals $G_1, \dots, G_M$.
Suppose that $\bm X \sim G$. We call the distribution of $\mathfrak{F}_G(\bm X)$ the copula of $G$, which we write as $C_G$.

Let $F$ be the CDF of a $M \times 1 $ random vector with marginals $F_1, \dots, F_M$. If
at least one of the marginals $F_i$ does not coincide with a marginal $G_i$, then
$\mathfrak{F}_F(\bm X) $ is a distribution on $(0,1)^M$, but it is not a copula. However, corollary~\ref{corr: equivalence of copulas} shows
that if $\bm X \sim G$, then the copula of $\mathfrak{F}_F(\bm X) $ is the copula of $G$.

\begin{corollary} \label{corr: equivalence of copulas}
If $\bm X \sim G$, then the copula of $\mathfrak{F}_F(\bm X)$ is $C_G$, which is the copula of $G$.
\end{corollary}

\begin{theorem} \label{thm: copula approx by mixture}
  Let $C$ be some arbitrary $M$-dimensional absolutely continuous copula with
  density $c$ and let $f$ be an approximation of $c$ as in theorem~\ref{corr: copula approx by mixture}. Then the copula of $f$ is also an approximation of $c$ in the $L_1$ norm.
\end{theorem}

\subsection{Universal approximation of multivariate distributions\label{SS: univ approx of general distns}}
Given that the approximating mixtures
\[\sum_{r = 1}^R \pi_r
\frac{\phi_{\mu_r, \sigma_r \bm I_M} \circ \mathfrak{F}_H^{-1}}{\prod_{i=1}^M h_i \circ H_i^{- 1}}
\]
can be transformed using $\mathfrak{F}_H^{-1}$ into normal mixtures, all the machinery that
exists for estimating finite normal mixtures (for example
{\cite{fruhwirth2006finite}}) or infinite normal mixtures (see
{\cite{kalli2011slice}} or \ {\cite{ishwaran2011gibbs}}) can be readily used
from a Bayesian perspective to obtain universal approximations to multivariate distributions.

The next section provides a simple simulated illustration of how such an approximation is implemented.

\subsection{Simple illustration of the universal approximation properties\label{SS: simple illustration}}

We now give a detailed example that illustrates how to apply our methodology and shows that a copula based multivariate approximation can be much more parsimonious
than that obtained by directly fitting a mixture of normals.

Consider the two-dimensional random vector $(X, Y)$ having joint distribution
\[ f (x, y)  = f_X (x) f_Y (y) c (F_X (x), F_Y (y)), \]
where each marginal is a mixture of univariate normal random variables,
\begin{align}
 f_X (x) & = f_Y (x) = \sum_{r = 1}^{R} \pi_{r} \phi_{\mu_{r},
   \sigma_{r}} (x_{})\notag
   \intertext{and the copula is}
 C (u, v) & = u^{1 - \alpha} v^{1 - \beta}  [u^{- \theta \alpha} + v^{- \theta
   \beta} - 1]^{- \frac{1}{\theta}},  \label{eq: example copula}
\end{align}
where $\alpha, \beta \in (0, 1)$ and $\theta > 0$

\begin{sloppypar}
We chose the following settings and parameter values: $f_X = f_Y$, $R = 6$, $\mu_{1:6} =\{- 9, - 5.4, -
1.8, 1.8, 5.4, 9\}$ and $\sigma_{r} = 1 / \sqrt{10}$, $\pi_r = 1 / R$,
$\alpha = 3 / 4$, $\beta = 1 / 2$ and $\theta = 20$.
\end{sloppypar}

We chose this example as it demonstrates the versatility  and power of our method because it is difficult to estimate
the density $f(x,y)$ with the usual estimators.
First, it is very difficult to
approximate the copula \eqref{eq: example copula} using either mixtures of Archimedean copulas or mixtures of
elliptical copulas because it is neither radially symmetric nor
exchangeable (see Sections~\ref{S: mixtures of archimedean copulas} and \ref{S: mixtures of elliptical copulas}).
This can also be checked directly from Figure~\ref{fig: example},
which shows a draw of 1000 points from the copula~\eqref{eq: example copula}.
Second, it is difficult  to approximate the joint distribution of
$X$ and $Y$ directly using mixtures of bivariate normals as
we chose the marginal distributions so that a mixture of bivariate normal densities would require up to 36 components; see
Figure~\ref{fig: example} for a plot of $n = 1000$ points from $f(x,y)$. Using our approach, a
six-component mixture of normals is enough to approximate each margin and, as
we show empirically,  a three-component mixture of normals is sufficient to
approximate the copula.  See Table~\ref{tab: marg lik estimates} for marginal likelihood estimates for 2 to 5 components
for the mixture in our approach and figure~\ref{fig: example} for an MCMC draw from that mixture.

\begin{figure}
\center
\begin{tabular}{@{}ll@{}}
\resizebox{150px}{150px}{\includegraphics{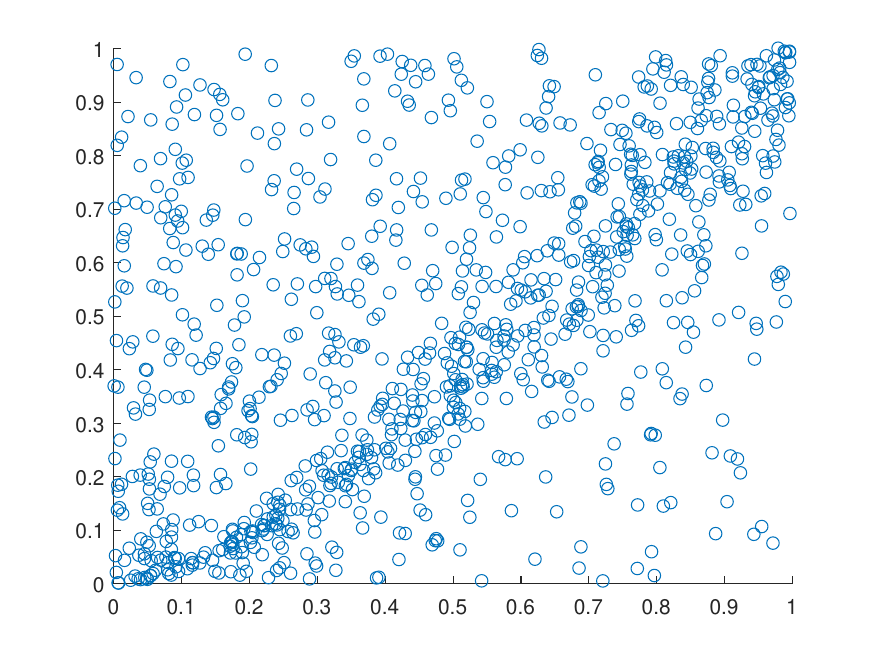}}
  &
  \resizebox{150px}{150px}{\includegraphics{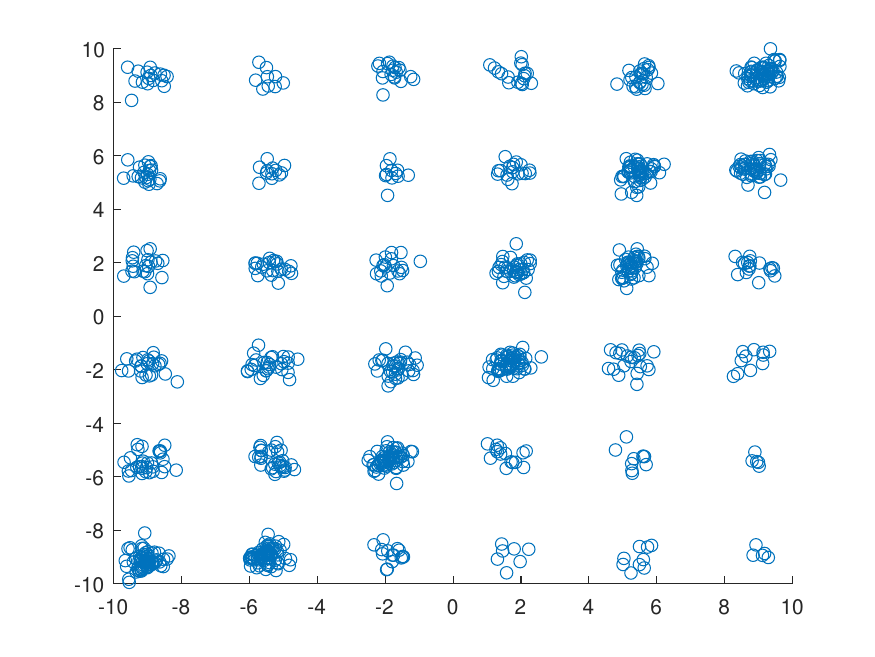}}
  \end{tabular}
   \resizebox{400px}{400px}{\includegraphics{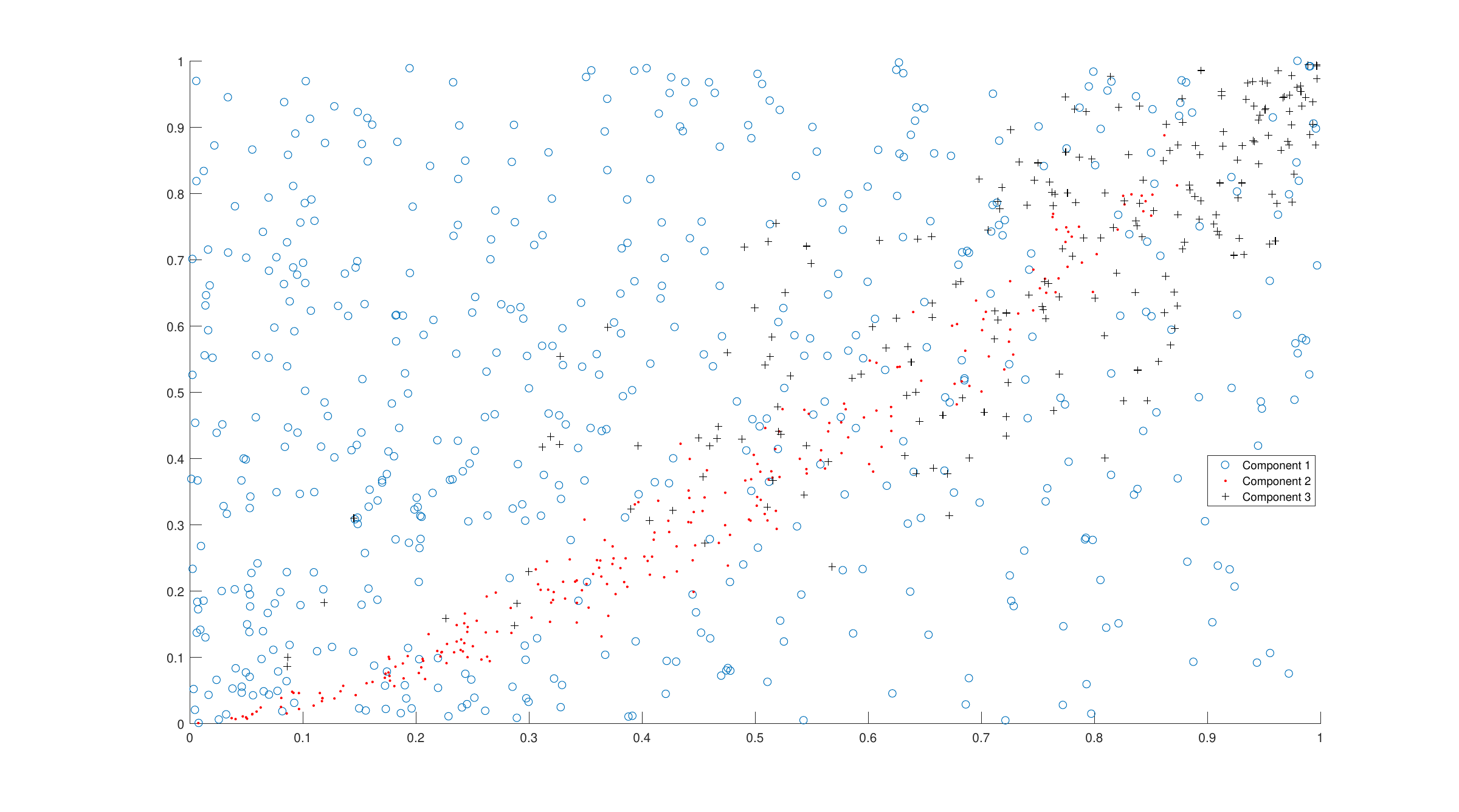}}
  \caption{top left panel: An example of a draw from the non-exchangeable copula; top right panel: Original data; bottom left panel:  An MCMC draw of the mixture indicators. Each of the components is
  labelled with a different color\label{fig: example} }
\end{figure}
\begin{table}[!h]
  \center
  \begin{tabular}{|c|c|c|c|c|}
    \hline
    $R$ & 2 & 3 & 4 & 5\\
    \hline
    Marg. Lik. & $- 2613.1$ & $- 2609$ & $-2623.8$ & $- 2637.6$\\
    \hline
  \end{tabular}
  \caption{Marginal likelihood estimates for different models computed using
  bridge sampling.\label{tab: marg lik estimates}}
\end{table}

We carried out the above analysis as follows.
\begin{enumerate}
  \item Draw a random sample $(x_i, y_i)$ for $i = 1,
  \ldots, n=1000$,  from the above model.  Appendix~\ref{app: drawing from example copula} shows how to
  draw from  the copula~\eqref{eq: example copula}.

  \item Estimate the marginal distributions using the empirical distribution
  functions $F_{X, n}$ and $F_{Y, n}$. Compute $z_i = \Phi^{- 1} \circ F_{X,
  n} (x_i)$ and $w_i = \Phi^{- 1} \circ F_{Y, n} (y_i)$ for $i = 1, \ldots,
  n$, where $\Phi$ is the CDF of the standard normal distribution.

  \item Approximate the joint distribution of $(z_i, w_i)$ using a mixture of
  multivariate normals.

  \item Apply $\Phi$ to recover the copula approximation.
\end{enumerate}
Step 3 can be carried out using any standard approach for estimating mixtures
of multivariate normals, and we applied textbook techniques from
\cite{fruhwirth2006finite}.

We adopt a Bayesian paradigm, but a
similar procedure can be carried out using a more classical approach. We repeated the model fitting
 for models with a different number of
components and chose the model with the highest marginal likelihood.
Independent multivariate normal and inverse Wishart priors are placed on the
parameters of the mixture components as described in Subsection 6.3.2 in
\cite{fruhwirth2006finite}. We then ran the MCMC algorithm 6.2 from
\cite{fruhwirth2006finite} for each model and its marginal
likelihood was computed by bridge sampling. All the MCMC computations were
carried out using the \texttt{bayesf\_version\_2.0} Matlab package
by \cite{fruhwirth2006finite}.

The  marginal likelihood estimates in Table~\ref{tab: marg lik estimates} were obtained by estimating
the copula density model \[\sum_{r = 1}^R \pi_r
    \frac{\phi_{\bm \mu_r, \sigma_r I_M} \circ \mathfrak{F}_H^{- 1}}{ \prod_{i=1}^M h_i \circ H_i^{- 1} }, \]
 specified by $H_1, \ldots, H_M$
standard normal distribution functions. This is valid because we use the same estimates of the marginals of $x$ and $y$ for all the fitted mixture models.

\section{Asymptotic Properties \label{S: asymptotic properties} }

We use the framework of \cite{massart2007concentration} to illustrate some of the asymptotic properties of estimators based on the mixtures constructed by our approach. We will first characterize the bracketing entropy numbers of the family of approximating mixtures by building on the results by \cite{maugis2011non}, then we will use those entropy numbers to apply Massart's framework and derive a non-asymptotic concentration inequality that will show contraction rates of Hellinger distance. Coupled with the approximation results from the previous section, the concentration inequality can be used to derive consistency of our estimators.

We will start by deriving results regarding the complexity of the family of mixture approximation. For that purpose, we will introduce some new notations. Denote by $\mathcal{F}$ the family of components of the mixture approximations, that is

\[ \mathcal{F} = \left\{ \frac{\phi_{\bm \mu, \bm \Sigma} \circ \mathfrak{F}_H^{- 1}}{ \prod_{i=1}^M h_i \circ H_i^{- 1} }; \bm \mu \in [- a, a]^M, 
 \bm \Sigma = \text{diag}(\sigma_1^2,...,\sigma_M^2), \sigma_1^2,...,\sigma_M^2 \in [\lambda_m, \lambda_M] \right\} \]
 
where $a>0$ and $\lambda_M>\lambda_m>0$. Notice, that we choose here to use diagonal matrices instead of scaled identity matrices for the covariance matrices of the components. The following results do not change by replacing $\Sigma$ with $\sigma I_M$.

Further, denote by $\mathcal{S}_R$ the family of mixtures with components in $\mathcal{F}$, that is

\[ \mathcal{S}_R = \left\{\sum_{r=1}^{R} \pi_r f_r; f_r \in \mathcal{F}, (p_1,...,p_r,...,p_R) \in \mathcal{P}_{R-1} \right\} \]

where $\mathcal{P}_{R-1}$ is the $(R-1)$-dimensional simplex, and where $f_r$ obviously are densities of the form $\frac{\phi_{\bm \mu_r, \Sigma_r} \circ \mathfrak{F}_H^{- 1}}{ \prod_{i=1}^M h_i \circ H_i^{- 1} }$. 

We are interested in characterizing the complexity of the family $\mathcal{S}_R$ through bracketing entropy in order to derive further asymptotic properties. 

If one defines the Hellinger distance by $d_H$ and let $\varepsilon>0$, then the bracketing entropy $\mathcal{H}_{[]}(\varepsilon,\mathcal{S}_R,d_H)$ is the logarithm of $\varepsilon$-bracketing covering number $\mathcal{N}_{[]}(\varepsilon,\mathcal{S}_R,d_H)$ with respect to $d_H$ defined as the minimum number of brackets $[c_l,c_u]$ covering $\mathcal{S}_R$, that is, given $\varepsilon$, there exists $N=\mathcal{N}_{[]}(\varepsilon,\mathcal{S}_R,d_H)$ different functions $c_{l1},...,c_{lN}$ and $c_{u1},...,c_{uN}$ such that $d_H(c_{lj},c_{uj})<\varepsilon$ for every $j=1,...,N$ and furthermore, for every $c \in \mathcal{S}_R$, there exists $j$ such that $c_{lj}<=c<=c_{uj}$. See e.g. \cite{van1996weak} or \cite{kosorok2008introduction} for further details.

\begin{proposition} \label{entropy numbers}
	For any $\varepsilon \in (0,1] $, and given the constant $K=\frac{5}{8} \left( 1 - 2^{- \frac{1}{4}} \right)$, then
	
	\[ \mathcal{N}_{[]}(\varepsilon,\mathcal{F},d_H) \leq 2^{\frac{11 M }{2}} M ^{2 M }  \left(a \frac{\lambda _M}{\lambda _m} \sqrt{\frac{1}{c \lambda _m}}\right){}^M  \left(\frac{1}{\varepsilon }\right)^{2 M } , \]
	which directly implies the $\varepsilon$-bracketing entropy $\mathcal{H}_{[]}(\varepsilon,\mathcal{S}_R,d_H)$ can be bounded above by the sum of a constant $K=K(R,M,a,\lambda_m,\lambda_M)$ and a term involving $\log(1/\varepsilon)$, that is
	
	\[  \mathcal{H}_{[]}(\varepsilon,\mathcal{S}_R,d_H) \leq K + R(2M+1)\log \left( \frac{1}{\varepsilon} \right)  \]

\end{proposition}

In this section, we show the results under the setting of known margins. Assume we have an i.i.d. sample $(\bm u_1, \ldots, \bm u_n)$ from a copula that has no singular component, and that has density $c_0$ with respect to Lebesgue's measure and where $\bm u_i \in \mathbb{R}^M$ for $i = 1, \ldots, n$. Let  \[ \text{KL} (c_0, c) = \int \log \left( \frac{c_0 (x)}{c (x)} \right) c (x)   \mathd x, \] be the Kullback-Leibler divergence between the the true copula $c_0$ and copula $c$. Indexing models in $\mathcal{S}_R$ by the number of components in the mixture $R$ and the dimension $M$, we are interested in estimators $\hat{c}$ minimizing the penalized contrast function
		\[  \text{crit}(R,M)=\gamma_n(c) + \text{pen}(R) \]
with respect to $c$ where $\gamma_n (c) = - \frac{1}{n} \sum_{i = 1}^N \log (c (\bm u_i))$, and where $\text{pen}(R,M)$ is a penalty function depending on $R$ and $M$. Optimizing over $c \in \mathcal{S}_R$ and using \cite{massart2007concentration} theorem 7.11 as illustrated in \cite{maugis2011non} theorem 1.1, we can prove the following result. 

\begin{theorem}\label{oracle inequality}
	Given two absolute constants $K_1$ and $K_2$, given the constant $ K=K(M,a,\lambda_m,\lambda_M) $ and given the penalty function satisfying the inequality
	\[  \text{pen}(R)  \geq K_1 \frac{R(2M+1)}{n}  \left( 1 + 2 K^2   + \log \left(  \frac{1}{1 \wedge \frac{R(2M+1)}{n} K^2}   \right)  \right),  \] there exists a model given by $\hat{R}$ that minimizes
				\[  \text{crit}(R)=\gamma_n( \hat{c}_{\hat{R}} ) + \text{pen}(R) \]
	
	and the estimator $\hat{c}_{\hat{R}}$ satisfies the inequality
	
	\[  E \left[d^2_H (c_0, \hat{c}_{\hat{R}}) \right] < K_2 \left[\inf_R  (\text{KL} (c_0, \mathcal{S}_R) + \text{pen} (R)) + \frac{1}{n} \right]  \]
	
	where $\text{KL} (c_0, \mathcal{S}_R) = \inf_{c \in \mathcal{S}_R} \text{KL}(c_0,c) $.

\end{theorem}

The theorem can be specialized to hold over a compact set that is a proper subset of $[0,1]^M$ in case $c_0$ is unbounded.

\section{Mixtures of Archimedean copulas\label{S: mixtures of archimedean copulas}}
Section \ref{SS: chracterization} summarizes some properties of Archimedean copulas and their mixtures. These properties are
used in section~\ref{SS: approximation Arch}  to derive some approximation
properties of mixtures of Archimedean copulas.

\subsection{Characterization of Archimedean copulas\label{SS: chracterization}}

Let $G$ be an Archimedean copula, that is a copula of the form
\begin{equation}
  G ({\bm {u}}) = \varphi \left( \sum_{m = 1}^M \varphi^{- 1} (u_m)  \right), \quad \text{where} \quad {\bm { u}} = (u_1, \dots, u_M),
 \label{archimedeancopula}
\end{equation}
and $\varphi$ is a completely monotone function. The stochastic
representation of an Archimedean copula  asserts that if a vector of uniform
random variables $\bm U = (U_1, \ldots, U_M)$ is distributed according to some
Archimedean copula distribution, then there exists a random variable $D$
with positive support such that
\begin{equation}
  \varphi (t) = E (e^{- Dt}) \label{stochasticrepresentation}
\end{equation}
and such that $U_1, \ldots,, U_M$ are conditionally independent given $D$. See for instance \cite{hofert2011efficiently} for more details about sampling Archimedean copulas and for additional references.

This means that given either the functional form in equation
(\ref{archimedeancopula}) or the stochastic representation in equation
(\ref{stochasticrepresentation}) based on $D$, the distribution of $U_1,
\ldots, U_M$ is exchangeable. That is, given any permutation $\sigma$ of the
set $\{ 1, \ldots, M \}$, the distributions of $U_1, \ldots, U_M$ and
$U_{\sigma (1)}, \ldots, U_{\sigma (M)}$ are identical.

If $G_1, \ldots, G_R$ are Archimedean copulas, then
$ G (\bm u ):= \pi_1 G_1 (\bm u ) + \cdots + \pi_R G_R (\bm u ) $
is a mixture of Archimedean copulas, and it is immediate that it is
exchangeable.

For additional details about the construction of Archimedean copulas, see for instance
\cite{mcneil2009multivariate}, {\cite{hofert2011efficiently}},
{\cite{mai2012simulating}} or {\cite{joe2014dependence}}.

\subsection{Approximation properties of a mixture of Archimedean copulas.  \label{SS: approximation Arch}}

Proposition~\ref{prop: non excha archi copula} shows that a mixture of Archimedean copulas $G$ is incapable of
approximating arbitrarily well any copula $C$ that is not exchangeable. To construct
an example, we need to find a copula that is not exchangeable. Given that $G$
is exchangeable, we need to find points in $(0, 1)^M$ that are separated
by $C$,  but not by $G$.

There are many non-exchangeable copulas. For instance, a typical Gaussian
copula in three dimensions or more is non-exchangeable. Even in 2 dimensions,
{\cite{durante2009construction}} constructs several bivariate
copulas that are non-exchangeable.

\begin{proposition} \label{prop: non excha archi copula}
  Let $c$ be the copula density of some non-exchangeable random vector.
  Then
  there exists an $\varepsilon > 0$ such that for all $R \in \mathbb{N}$,
  every $\pi = (\pi_1, \ldots, \pi_R)$ in the $R$-simplex and every possible
  set of Archimedean copula densities $g_1, \ldots, g_R$,
  \[ \left\| c - \sum_{r = 1}^R \pi_r g_r \right\| \geqslant \varepsilon > 0,
  \]
  for the $L_\infty$ norm. If $c$ is also continuous, then the result also holds for the
  $L_1$ norm.
\end{proposition}

In proposition~\ref{prop: non excha archi copula} and below, we define the $L_\infty $ norm for
$f:(0,1)^M\rightarrow \mathbb{R}$ as $\| f\|:= \sup_{\bm u \in (0,1)^M}|f(u)| $.

\section{Mixtures of Elliptical copulas\label{S: mixtures of elliptical copulas}}

An elliptical copula is the copula of a random vector that has an elliptical
distribution. Section~\ref{SS: characterization of elliptical copulas} describes some of the properties of elliptical copulas and
section~\ref{SS: approx by mixtures of elliptical copulas}  derives some of their
 approximation properties.

\subsection{Characterization of elliptical copulas}\label{SS: characterization of elliptical copulas}

\subsection*{Definition of an elliptical copula} An elliptical copula is the copula of a vector random variable $\bm X$ that is
elliptically distributed. The $M$-dimensional random vector $\bm X$ is called elliptically distributed with
location $\bm \mu \in \mathbb{R}^M$ and scale matrix $\bm \Sigma$, where $\bm \Sigma$ is
a symmetric and positive semi-definite matrix, if
\[ \bm X \xequal{d} \bm \mu + \bm {A}' \bm {Y}, \]
where $\bm Y$ is some spherically distributed random vector and $\bm \Sigma = \bm A' \bm A$.
$\bf \Sigma$ also admits the variance-correlation decomposition
\[ \bm \Sigma = \bm {S R S}, \]
where $\bm R$ is a correlation matrix and $\bm S$ is a diagonal matrix having standard
deviations on the main diagonal.

\subsection*{Example: Gaussian copula}
The simplest example of an elliptical distribution is the multivariate normal, $\bm X
\sim \mathcal{N} (\bm \mu, \bm \Sigma)$. In this case the $L_2$ norm of $\bm Y$ is
distributed as a chi-squared with $M$ degrees of freedom.

Let $F$ be the cdf of an $M \times 1 $ random vector, with marginal cdf's $F_1, \dots, F_M$.
   We define $\mathfrak{F}_F: \mathbb{R}^M \rightarrow (0,1)^M $ as $\mathfrak{F}_F(\bm x):= \left (F_1(x_1), \dots, F_M(x_M) \right )$
   In particular, if $F_X(\bm x)$ is the CDF of $\bm X \sim \mathcal{N}(\bm \mu, \bm \Sigma)$, then $\mathfrak{F}_{F_X}$
     is distributed as a Gaussian copula with correlation  matrix $\bm R$.

We now show that the distribution of $\mathfrak{F}_{F_X}$ is invariant under linear transformations, where, without loss of generality, we assume
that $\bm \Sigma$ is positive definite. Let $f_X(\bm x) $ be the density function of $\bm X$.
Then, the characteristic function of $\mathfrak{F}_{F_X}$ is
\begin{eqnarray*}
  \mathfrak{M} (\bm t) & = & \int_{\mathbb{R}^M} {\rm e}^{i \langle \mathfrak{F}_{F_X}, \bm t \rangle} f_X(\bm x)  \mathd \bm x\\
  & = & \int_{\mathbb{R}^M} {\rm e}^{i \langle \mathfrak{F}_{F_Z}, \bm t \rangle} f_Z(\bm z)
  \mathd \bm z
\end{eqnarray*}
by the change of variable $\bm x = \bm{Sz} + \bm \mu$. Clearly,
only the correlation matrix $\bm R$ is identified, but not the location and
scale parameters $\bm \mu$ and $\bm S$.

\begin{remark} \label{rem: identif of elliptic copula}
We note that the identification properties derived above for a Gaussian copula extend immediately
to any elliptical copula.
\end{remark}

\subsection{Approximation properties of mixtures of elliptical copulas \label{SS: approx by mixtures of elliptical copulas}}

Mixtures of elliptical  copulas cannot always approximate an arbitrary copula
$C$. The key concept we will use to construct a counterexample of why the
approximation breaks down is that of radial symmetry.

\subsection*{Definition of radial symmetry}
Let $\bm {1}_M$ be the vector of ones in the unit $M$-cube.
A copula $G$ is said to be radially symmetric if given any $\bm u \in (0, 1)^M$,  then $G (\bm u) = G (\bm {1}_M - \bm u) $.

If a copula $C$ is not radially symmetric, then it cannot be approximated by a mixture of radially symmetric
copulas because any such mixture would also be radially symmetric and
thus cannot separate points on $(0, 1)^M$ that are separated by $C$.

Proposition~\ref{Prop: radial symmetry}  shows that any elliptical copula is radially symmetric and
hence so is a mixture of elliptical copulas.

\begin{proposition} \label{Prop: radial symmetry}
\begin{enumerate}
\item [(i)]
 If $G$ is an elliptical copula, then $G$ is radially symmetric.
 \item [(ii)]
Suppose that  $ G =  \pi_1 G_{1, \bm{R}_1}+ \cdots +  \pi_R G_{1, \bm {R}_R}$
 is a mixture of radially symmetric copulas. Then, $G$ is radially symmetric.
 \end{enumerate}
\end{proposition}

  \begin{remark}\label{re: remark on radial symmetry}
  It is interesting to note that
  radial symmetry  fails in the case of a finite mixture of
  elliptical distributions because when each component of the mixture is
  radially symmetric around a different point $\bm \mu_r \in \mathbb{R}^M$, then
  the radial symmetry fails overall unless the following equalities hold
  $\bm \mu_1 = \bm \mu_2 = \cdots = \bm \mu_R = \bm \mu$ for some $\bm \mu \in \mathbb{R}^M$.
  \end{remark}

Proposition~\ref{prop: elliptic copulas and radial symmetry} shows that a mixture of Gaussian copulas is incapable of approximating an arbitrary copula
$C$. All we need for a counterexample is a copula that is not radially symmetric.
\begin{proposition} \label{prop: elliptic copulas and radial symmetry}
  Let $c$ be the copula density of a random vector that is not radially symmetric.
  Then there exists $\varepsilon > 0$ such that for all $R \in \mathbb{N}$,
  every $\pi = (\pi_1, \ldots, \pi_R)$ in the $R$-simplex and every possible set of
  elliptical densities $g_1, \ldots, g_R$,
  \[ \left\| c - \sum_{r = 1}^R \pi_r g_r \right\| \geqslant \varepsilon > 0,
  \]
  for the $L_\infty$ norm. If $c$ is continuous, then the result also holds for the $L_1$ norm.
\end{proposition}

\section{Empirical illustration: Application to the dependence between large financial firms\label{S: empirical illustrations}}
We  follow Section 4 of \cite{OhPpatton2013} using the same data and fit our
flexible approximation to a copula that models the dependence between seven large financial  institutions (Bank
of America, Citigroup, Bank of New York, Goldman Sachs, J.P. Morgan, Wells
Fargo and Morgan Stanley) over the period 2000-12-25 to
2011-01-05 with a total of 2521 observations. We show that our approach provides a better fit and is more parsimonious.

Let $r_{i,t},i=1, \dots,7$, be the return for the $i$th firm at time $t$, and $r_{m,t}$ the return on the S\&P 500 index at time $t$.
\cite{OhPpatton2013} fit the following model to the data using simulated method of moments estimation.
\begin{eqnarray*}
  r_{i, t} & =  \phi_{0, i} + \phi_{1, i} r_{i, t - 1} + \phi_{2, i} r_{m, t
  - 1} + \varepsilon_{i, t}, \quad  \varepsilon_{i, t}  =  \sigma_{i, t} \eta_{i, t}\\
  \sigma_{i, t}^2 & =  \omega_i + \beta_i \sigma_{i, t - 1}^2 + \alpha_{1, i}
  \varepsilon_{i, t - 1}^2 + \gamma_{1, i} \varepsilon_{i, t - 1}^2
  \mathbbm{1} (\varepsilon_{i, t - 1} \leqslant 0)\\
  & + \alpha_{2, i} \varepsilon_{m, t - 1}^2 + \gamma_{2, i}
  \varepsilon_{m, t - 1}^2 \mathbbm{1} (\varepsilon_{m, t - 1} \leqslant 0)\\
  r_{m,t} & =  \phi_{0m} + \phi_{1m} r_{m,t-1} + \varepsilon_{m,t} , \varepsilon_{m,t} = \sigma_{m,t}\eta_{m,t} \\
  \sigma_{m,t}^2 & = \omega_m + \beta_m \sigma_{m, t - 1}^2 + \alpha_{1, i}
  \varepsilon_{m, t - 1}^2 + \gamma_{m} \varepsilon_{m, t - 1}^2
  \mathbbm{1} (\varepsilon_{m, t - 1} \leqslant 0)\\
\end{eqnarray*}
where $i = 1, \ldots, 7$, $t = 1, \ldots, 2521$ and where $\eta_{i, t} \sim
\mathcal{N} (0, 1)$. \cite{OhPpatton2013} then  estimate the $\eta_{i,t}$ and fit a Gaussian copula to $\wh \eta_t = (\wh \eta_{1t}, \dots, \wh \eta_{7t})^{\transp}$
to study the joint dependence of the stock returns.

Let $u_{i, t} = \Phi (\eta_{i, t})$, where $\Phi$ is the
cumulative distribution function of a standard normal (alternatively let
$v_{i, t} = F_T (\eta_{i, t})$, where $F_T$ is the empirical distribution
function). We then fitted two kinds of models to $\tmmathbf{u}$ (or $\tmmathbf{v})$
\begin{enumerate}
  \item [(i)] Mixture I. A mixture of Gaussian copulas $\sum_{r = 1}^R \pi_r c (\cdummy,
  \tmmathbf{C}_r)$, where the $\tmmathbf{C}_r$ are correlation matrices.

  The number of parameters in this model is $R \frac{M (M - 1)}{2} + R - 1$
  where $M = 7$.

  \item [(ii)] Mixture II. A mixture of normals to $\Phi^{- 1} (v_{i, t})$ and then recover the copula.

  The number of parameters in this model is $R \frac{M (M + 3)}{2} + R - 1 - 2
  M$. The $-2M$ term arises in the last expression occurs because when in fitting a copula, $M$ means and $M$ variances are not determined.
\end{enumerate}

We note that Mixture I with $R = 1$ is the \cite{OhPpatton2013} approach.

\begin{table}[!ht]
   \caption{BIC values for the mixture of Gaussian copula and approximating mixture for estimating the distribution of
   $\eta_t$ flexibly.\label{table: BIC values for two methods fin data}}
   \begin{center}
 \begin{tabular}{ccc}
    \# components & Mixture of Gaussian copulas&
     Approximating mixture\\
     1 & 36374 & 36374\\
     2 & $- 12912$ & $- 17146 $\\
     3 & $- 12786$ & $- 16930$\\
     4 & $- 12660$ & $- 16714$
   \end{tabular}
   \end{center}
 \end{table}

Table~\ref{table: BIC values for two methods fin data} reports the BIC values for each of the 4 models for each of Mixture I and Mixture 2.
The table shows that a Gaussian copula provides an inadequate fit and the mixture of $R = 2$
Gaussian copulas provides the best fit if we use a mixture of Gaussian copulas. The table also shows that the best approximating mixture has two components (BIC of $-17146$) and provides a far better fit than the best mixture of Gaussian copulas (BIC of $-12912$).  If we take all models as equally likely, and use $\exp(-\frac12 BIC)$ as an estimate of the marginal likelihood of each model under flat priors, then the ratio of the posterior probability of the best approximating model to the Gaussian copula models is
$\exp(26670)$ and the ratio of the best approximating model to the best approximation by a mixture of Gaussian copulas is $\exp(2117)$.

\section{Conclusion \label{S: conclusion}}
Our article provides fundamental tools for approximating any copula arbitrarily well and uses these to propose a practical family of mixtures to provide such an
approximation. We can then use this approximation to construct a practical  copula-based approach for approximating any multivariate distribution arbitrarily well.
Such a copula approach for universally approximating multivariate distributions is attractive as it allows us to control the degree of approximation of the marginal distributions as well as providing a flexible way of approximating the joint dependence. Furthermore, the approach is easy to implement and satisfies good asymptotic properties.
Thus, our approach can provide an attractive alternative to approximating multivariate distributions by a mixture of normals.
We also study the approximation properties of mixtures of Gaussian copulas or mixtures of Archimedean copulas and show that neither family of mixtures
can  approximate a general copula arbitrarily well.

Furthermore, the universal approximation results proved in this paper are theoretical and of a probabilistic/analytic nature, and thus are essential for further statistical analysis of the problem. In fact, they constitute standard density results (like showing that a continuous function under certain assumptions can be approximated by polynomials and splines) that are a cornerstone for all sieve-based non-parametric estimation techniques and are a pre-requisite for further analysis. This means that they bring the whole machinery of mixture modeling to bear on the problem of non-parametric estimation of copulas. Given that the results show it is legitimate to use our mixture based model under certain conditions to approximate copulas, then the standard mixture machinery can then be legitimately used to estimate that model and hence the copula. 
In particular, if one wants to use nonparametric copula estimation using sieves built from mixtures, or Dirichlet process mixtures in a Bayesian setting, then the results are both a pre-requisite and foundational. See, for example, the way sieves are constructed in say \cite{shen1997methods} or \cite{chen2007large}.

\section*{Acknowledgement}
Robert Kohn's research was partially supported by an Australian Research Council grants  DP150104630 and CE140100049.

\appendix
\section{Proofs \label{app: proofs}}

\begin{proof}[\textbf{Theorem~\ref{bacharoglu_theorem}}]
  The proof follows from {\cite{bacharoglou2010approximation}}.  See in that paper theorem 2.4 for the compact
  support approximation case, corollary 2.5(2) for the $L_1$ approximation case, and
  corollary 2.5(3) for the $L_{\infty}$ approximation case.
\end{proof}

\begin{proof}  [ \textbf{Theorem~\ref{thm: one dim approx} }]
We first prove the theorem for the $L_1$ norm.
Let $U$ be a random variable with density $g$. Applying the
  transformation $X= H^{- 1} (U)$ yields an absolutely continuous random variable
  with support on the whole real line with density $f = h \cdot (g
  \circ H^{- 1})$. Furthermore, $f$ is clearly both continuous and bounded. We
  can apply theorem \ref{bacharoglu_theorem} to get the following
  approximation property.

  For every $\varepsilon > 0$, there exists a sequence $(\alpha_n)_{n \in
  \mathbb{N}}$ in $\mathcal{A}^+$ and an integer $R$ such that in the
  enumeration $(\phi_n)_n$ specified by $\phi_{\mu,\frac{1}{ k}} (x )$, where
  $k \in \mathbb{N}$, $\mu \in \mathbb{Q}$, the convex combination
  \begin{align*}
   \sum_{r = 1}^R \pi_r  \phi_r,
   \quad \text{with} \quad \pi_r:=\frac{\alpha_r}{\sum_{r' = 1}^R \alpha_{r'}},
   \end{align*}
  is arbitrarily close to $f$. Denoting the normalized weights by $\pi_r$ and
  the normal densities parameters in the enumeration by $\mu_1, \ldots, \mu_R,
  \sigma_1, \ldots, \sigma_R$ yields the required result.
  \[ \left\| f - \sum_{r = 1}^R \pi_r \phi_{\mu_r, \sigma_r } \right\| <
     \varepsilon . \]
  Finally, applying the transformation $H$ yields
  \begin{align*}
    \varepsilon & >  \left\| f - \sum_{r = 1}^R \pi_r \phi_{\mu_r, \sigma_r }
    \right\|_1\\
    & =  \int_{- \infty}^{\infty} \left| f (x) - \sum_{r = 1}^R \pi_r
    \phi_{\mu_r, \sigma_r } (x) \right| \mathd x\\
    & =  \int_0^1 \left| \frac{f \circ H^{- 1}}{h \circ H^{-1}} (u) - \sum_{r = 1}^R \pi_r
   \frac{ \phi_{\mu_r, \sigma_r } \circ H^{- 1}}{h \circ H^{-1} }  (u) \right|  \mathd u\\
    & =  \int_0^1 \left| g (u) - \sum_{r = 1}^R \pi_r  \frac{\phi_{\mu_r,
    \sigma_r } \circ H^{- 1}}{h \circ H^{-1} }  (u) \right|\mathd u\\
    & = \left\| g - \sum_{r = 1}^R \pi_r \frac{ \phi_{\mu_r, \sigma_r} \circ
    H^{- 1}}{h \circ H^{-1} } \right\|.
  \end{align*}
We now consider the $L_\infty$ case. Suppose the result does not hold for this case. Then, there exists an $\varepsilon > 0 $, such that
for any $R \in \mathbb{N}$, $(\pi_1,
  \ldots, \pi_R) \in \Delta_R$ (the $R$-simplex), $\mu_1, \ldots, \mu_R \in
  \mathbb{R}$ and $\sigma_1, \ldots, \sigma_R \in (0, \infty)$ such that
  \[ \left\| g - \sum_{r = 1}^R \pi_r \phi_{\mu_r, \sigma_r} \circ H^{-
     1}\right\|_\infty \geq  \varepsilon .  \]
     This implies that the result of theorem  does not hold for the $L_1$ norm as $g(\cdot)$ is continuous, providing a contradiction.
\end{proof}

\begin{proof} [ \textbf{Theorem \ref{corr: copula approx by mixture} }]
We prove the theorem for the $L_1$ norm.  The proof for the $L_\infty$ norm is similar to that in the proof
of theorem~\ref{thm: one dim approx}.
  Applying the inverse transformation $\mathfrak{F}_H^{- 1} :  [0, 1]^M \rightarrow
 \mathbb{R}^M$ yields an $\mathbb{R}^M$ random vector with density
  \[ f (\bm x) = c (\mathfrak{F}_H (\bm x)) \prod_{j = 1}^M h_j (x_j). \]
The function $f$ is trivially in
  $L_1$ (with respect to Lesbesgue measure) as it is the density of an absolutely continuous random vector.
  Furthermore $f \in L_{\infty}
  \cap C (\mathbb{R}^M)$ because it is bounded and continuous.
  Suppose $\varepsilon > 0$ is given. Applying theorem
  \ref{bacharoglu_theorem} to $f$, there exists a sequence $(\alpha_n)_{n \in
  \mathbb{N}}$ in $\mathcal{A}^+$ and an integer $R$ such that in the
  enumeration $(\phi_n)_n$ specified by $\phi_{\frac{1}{k}} (\bm x - \bm \mu)$ where
  $k \in \mathbb{N}$, $\bm \mu \in \mathbb{Q}^M$, the convex combination
\begin{align*}
 \sum_{r = 1}^R \pi_r  \phi_r , \quad \text{where} \quad \pi_r:= \frac{\alpha_r}{\sum_{r' = 1}^R  \alpha_r^\prime}.
 \end{align*}

\begin{sloppypar}
  If the mean vector and variance parameters corresponding to the
  enumeration are $\bm \mu_1, \ldots, \bm \mu_R, \sigma_1, \ldots, \sigma_R$, then
\end{sloppypar}

	\[ \left\| f - \sum_{r = 1}^R \pi_r \phi_{\bm \mu_r, \sigma_r \bm I_m} \right\| <
     \varepsilon . \]
  Explicitly writing the previous expression and applying the transformation
  $H$ yields
  \begin{eqnarray*}
    \varepsilon & > & \int_{\mathbb{R}^M} \left| f (\bm x) - \sum_{r = 1}^R \pi_r
    \phi_{\bm \mu_r, \sigma_r \bm I_m} (\bm x) \right| \mathd \bm x\\
    & = & \int_{(0, 1)^M} \left| c (\bm u) - q_R(\bm u ) \right| \mathd \bm u
  \end{eqnarray*}
  with $q_R(\bm u ) :=\sum_{r = 1}^R \pi_r
    \frac{\phi_{\bm \mu_r, \sigma_r \bm I_m} \circ \mathfrak{F}_H^{- 1}}{ \prod_{i=1}^M h_i \circ H_i^{- 1} } (\bm u)$.
\end{proof}

\begin{proof}[\textbf{Corollary~\ref{corr: marginals1}}]
We need the following standard result that is adapted from {\cite{devroye2012combinatorial}}.

Let $T$ be a Borel measurable mapping from $\mathbb{R}^M$ into $\mathbb{R}^L$ and let $f$ and $g$ be the density functions of two
arbitrary  $\mathbb{R}^M$ random vectors and $f_T$ and $g_T$ be respectively the densities of the mapped random vectors then
  \[ \| f - g \|_1 \geqslant \| f_T - g_T \|_1 \].
This is proved as follows. Let $f$ and $g$ be the densities of $X$ and $Y$ respectively.
  \begin{eqnarray*}
    \| f - g \|_1 & = & 2 \sup_{A \in \mathcal{B}(\mathbb{R}^M)} | \Pr \{ X \in A \} - \Pr
    \{ Y \in A \} |\\
    & \geqslant & 2 \sup_{A \in \mathcal{B}(\mathbb{R}^L)}| \Pr \{ T (X) \in A \}- \Pr \{
    T (Y) \in A \nobracket \} |\\
    & = & \int_{\mathbb{R}^M} | f_{T (X)} - g_{T (X)} | \mathd \mu,
  \end{eqnarray*}
  where the first line is Scheff{\'e}'s identity (theorem 5.1 in
  {\cite{devroye2012combinatorial}}) and the second line follows from theorem
  5.2 in the same reference.

A simple application of that result now yields our corollary. Consider the transformation $T(\bm u):= u_i$.
The proof for the $L_1$ norm now follows immediately. The proof for the $L_\infty$ norm follows because the marginals of $c$ are uniform.
\end{proof}

\begin{proof} [ \textbf{Corollary~\ref{corr: equivalence of copulas} }]
Let $\bm V = \mathfrak{F}_F(\bm X)$, let $F_V$ be the CDF of $\bm V$ with marginals $F_{V,i}$  and let $C_V$ be the copula of $\bm V$.
Let $W_i:= F_{V,i}(V_i), i=1, \dots, M$, $\bm W:=(W_1, \dots, W_M) $. Then,
\begin{align*}
w_i& = F_{V,i}(v_i) = G_i(F_i^{-1}(v_i) ) = G_i(F_i^{-1}(F_i(x_i)) )= G_i(x_i).
\end{align*}
The characteristic function of $\bm  W$ is
\begin{align*}
    \mathfrak{M} (\bm{t}) & =  \int_{\mathbb{R}^M} \exp \left (i \left (
    \sum_{i=1}^M t_i w_i \right ) \right )   g (\bm {x})
    \mathrm{d} \bm{x} = \int_{\bm{R}^M} \exp \left (i \left (
    \sum_{i=1}^M t_i G_i(x_i) \right ) \right )  g (\bm {x})
    \mathrm{d} \bm {x},
     \end{align*}
which is the characteristic function of $C_G$.
\end{proof}

\begin{proof} [ \textbf{ Theorem \ref{thm: copula approx by mixture}  }]

Assume there exists $f_{\varepsilon}$ such that $| c - f_{\varepsilon} | < \varepsilon$.

Let $f_1, \ldots, f_M$ and $F_1, \ldots, F_M$ be respectively the marginal
densities and distribution functions of $f_{\varepsilon}$ and let $f_{\pi} = f_1
\cdots f_M$ be the product of the marginal densities.

Let $v := v (u) = (F_1^{- 1} (u_1), \ldots, F_M^{- 1} (u_M))$ be the
point in $[0, 1]^M$ that results by applying the transformation $F_m^{-
1}$ to each coordinate $u_m$ of $u = (u_1, \ldots, u_M)$.

Thus we an write $f_{\varepsilon} (v) = f_{\pi} (u) c_f (u)$ where $c_f$ is
the copula density of $f_{\varepsilon}$.

It is well know that the Kullback-Leibler divergence satisfies the following
inequality with respect to the $L_1$ norm
\[ \| c_f - c \|_1^2 \leqslant 2 KL (c_f, c) \]
so that in order to bound the $L_1$ norm from above, we will try to bound the
Kullback-Leibler divergence
\begin{eqnarray*}
  KL (c_f, c) & = & \int \log \left( \frac{c_f (u)}{c (u)} \right) c_f
  (u) \mathd u\\
  & = & \int \log \left( \frac{c_f (u)}{c (u)} \times \frac{f_{\pi}
  (u)}{f_{\pi} (u)} \right) c_f (u) \mathd u\\
  & = & \int \log \left( \frac{f_{\varepsilon} (v)}{c (u)} \right) c_f (u)
  \mathd u - \int \log (f_{\pi} (u)) c_f (u) \mathd u\\
  & = & \int \log \left( \frac{f_{\varepsilon} (v)}{c (u)} \times
  \frac{f_{\varepsilon} (u)}{f_{\varepsilon} (u)} \right) c_f (u) \mathd u -
  \int \log (f_{\pi} (u)) c_f (u) \mathd u\\
  & = & \int \log \left( \frac{f_{\varepsilon} (u)}{c (u)} \right) c_f (u)
  \mathd u + \int \log \left( \frac{f_{\varepsilon} (v)}{f_{\varepsilon} (u)}
  \right) c_f (u) \mathd u - \int \log (f_{\pi} (u)) c_f (u) \mathd u
\end{eqnarray*}
so that
\[ KL(c_f, c) \leqslant \left| \int \log \left( \frac{f_{\varepsilon}
   (u)}{c (u)} \right) c_f (u) \mathd u \right| + \left| \int \log \left(
   \frac{f_{\varepsilon} (v)}{f_{\varepsilon} (u)} \right) c_f (u) \mathd u
   \right| + \left| \int \log (f_{\pi} (u)) c_f (u) \mathd u \right| \]

The first term can be made arbitrarily small by theorem 3. The last term
can be made arbitrarily small by corollary \ref{corr: marginals1}, as $f_{\pi} (u)$ can
be made arbitrarily close to 1. Finally, if we assume that the transformation
$H$ has a continuous density on $\mathbbm{R}$ (Assuming that $H \in C_1
(\mathbbm{R})$), then $f_{\varepsilon} (v)$ can be made arbitrarily close to
$f_{\varepsilon} (u)$ by the continuity of $f_{\varepsilon}$ as $u$ can be made
arbitrarily close to $v$ by corollary \ref{corr: marginals1}.

\end{proof}

\begin{proof} [\textbf{Proposition~  \ref{entropy numbers} }]
	
The proof follows closely the argument of theorem 2.1 \cite{maugis2011non}. For that argument to hold, we need to prove two additional results that are necessary to check for $\mathcal{F}$. As the argument requires the construction of a lattices over the parameters of $\mathcal{F}$, given that we are not dealing with multivariate normal densities, but with multivariate normal densities applied after some monotonic transformations, we need to check that the brackets constructed for the multivariate normal density case can be constructed in the same here, given both the multivariate normal density family and $\mathcal{F}$ share the same parameter space. 

The first result is that proposition C.1 in \cite{maugis2011non} dealing with an upper bound on the ratio of two normal densities also hold in the case of densities in $\mathcal{F}$ because the denominators in the densities expressions simplify and because $[0,1]^M \subset \mathbb{R}^M$.

The second result is that proposition C.3 in \cite{maugis2011non} dealing with the Hellinger distance between the upper and lower functions in the bracket can be computed in the same way. Please note that

\begin{eqnarray*}  d^2_H (\phi_{\tmmathbf{\mu}_1, \tmmathbf{\Sigma}_1}, \phi_{\tmmathbf{\mu}_2,  \tmmathbf{\Sigma}_2}) & = & 2 - 2 \int_{\mathbb{R}^M} \sqrt{\phi_{\tmmathbf{\mu}_1, \tmmathbf{\Sigma}_1} (\tmmathbf{x})  \phi_{\tmmathbf{\mu}_2, \tmmathbf{\Sigma}_2} (\tmmathbf{x})} hd  \tmmathbf{x}\\  & = & 2 - 2 \int_{[0, 1]^M} \sqrt{\frac{\phi_{\tmmathbf{\mu}_1,  \tmmathbf{\Sigma}_1} \circ \mathfrak{F}_H^{- 1} (\tmmathbf{u})}{\prod_{i =  1}^M h_i \circ H_i^{- 1} (u_i)} \times \frac{\phi_{\tmmathbf{\mu}_2,  \tmmathbf{\Sigma}_2} \circ \mathfrak{F}_H^{- 1} (\tmmathbf{u})}{\prod_{i =  1}^M h_i \circ H_i^{- 1} (u_i)}} d \tmmathbf{u}\\  & = & d^2_H \left( \frac{\phi_{\tmmathbf{\mu}_1, \tmmathbf{\Sigma}_1} \circ  \mathfrak{F}_H^{- 1}}{\prod_{i = 1}^M h_i \circ H_i^{- 1}},  \frac{\phi_{\tmmathbf{\mu}_2, \tmmathbf{\Sigma}_2} \circ \mathfrak{F}_H^{-  1}}{\prod_{i = 1}^M h_i \circ H_i^{- 1}} \right)
\end{eqnarray*}

The remainder of the proof for the calculation of the bracketing entropy for $\mathcal{F}$ proceeds in the construction in the lattice in exactly same way for multivariate normal densities (proof of theorem 2.2 in \cite{maugis2011non}).

The calculation of the entropy for $\mathcal{S}_R$ proceeds from the calculation for $\mathcal{F}$ and from theorem 2 in \cite{genovese2000rates}.
	
\end{proof}

\begin{proof} [\textbf{Theorem~  \ref{oracle inequality} }]

In the statement of the theorem, $K$ is given by
\[  K = \sqrt{\pi} + \sqrt{\log (18 \pi e^2)} + \sqrt{\log \left( a\sqrt{\frac{8}{c_1 \lambda_m}} \right)} + \sqrt{\log \left( 8\frac{\lambda_M}{\lambda_m} \right)} + \sqrt{\log \left( 9 \sqrt{2} M\right)}  \]

The remainder proceeds in exactly the same way as the proof of theorem 2.1 in \cite{maugis2011non}. In particular, as cited in the proof of that paper, the function $\Psi_R$ defined by $\Psi_R (\xi) = \xi \sqrt{R(2M+1)} \left\{K+ \sqrt{\log \left( \frac{1}{1 \wedge \xi} \right)} \right\}$ provides an upper bound for the entropy integral.
	
\end{proof}

\begin{proof} [\textbf{Proposition~\ref{prop: non excha archi copula} }]
We first prove the result for the $L_\infty$ norm.
  The densities $g_r$ being exchangeable means that for every permutation
  $\sigma$ of the set $\{ 1, \ldots, M \}$, we have the identity $g_r (\bm u
) = g_r (\bm u_\sigma)$, where $\bm u_\sigma:= (u_{\sigma(1)}, \dots, u_{\sigma(M)}) $.
 Taking convex combinations retains that symmetry. If we
  define $g : = \sum_{r = 1}^R \pi_r g_r$, then $g$ is also exchangeable,
  \begin{align*}
    g (\bm u ) & =  \sum_{r = 1}^R \pi_r g_r (\bm u )\\
    & =  \sum_{r = 1} \pi_r g_r (\bm u_\sigma) =  g (\bm u_\sigma).
  \end{align*}
  There exists at least one $\bm u \in (0, 1)^M$ and
  one $\sigma$ such that $c (\bm u ) \neq c (\bm u_{\sigma}) $ because $c$ is non-exchangeable,
  and hence $g$ is
  incapable of separating some points that $c$ is capable of separating. For
  one of these points, define for $\eta > 0$
  \begin{align*}
    \eta & =  | c (\bm u ) - c (\bm u_{\sigma})|\\
    & =  | c (\bm u ) - g (\bm u ) + g (\bm u ) - g (\bm u_{\sigma}) + g (\bm u_{\sigma}) - c (\bm u_{\sigma}) |\\
    & \leqslant  | c (\bm u ) - g (\bm u ) | + | c    (\bm u_{\sigma}) - g (\bm u_\sigma ) |  + | g     (\bm u ) - g (\bm u_{\sigma}) |\\
    & =  | c (\bm u ) - g (\bm u ) | + | c (u_{\sigma})  - g (\bm u )   |\\
    & \leqslant  2 \max \{ | c (\bm u) - g (\bm u ) |,
    | c (\bm u_{\sigma}) - g (\bm u_{\sigma }) | \}\\
    & \leqslant  2 \sup_{\bm u \in (0, 1)^M} | c (\bm u) - g (\bm u) | .
  \end{align*}
Hence,
  \begin{eqnarray*}
    \sup_{\bm u \in (0, 1)^M} | c (\bm u) - g (\bm u) | & \geqslant &  \eta/2 > 0.
  \end{eqnarray*}
  Choosing $\varepsilon = \eta / 2$ is sufficient to prove the proposition.

  Now all that is necessary is to find such an $\eta$,  given that we are
  constructing a counter-example. For $M = 2$, consider the copula (which is
  constructed using arguments in {\cite{durante2009construction}})
  \[ C (u, v) = u^{1 - \alpha} v^{1 - \beta} [u^{- \theta \alpha} + v^{-
     \theta \beta} - 1]^{- \frac{1}{\theta}}, \]
  where $\alpha, \beta \in (0, 1)$ and $\theta > 0$. We need to impose the
  constraint that either $\alpha \neq \frac{1}{2}$ or $\beta \neq \frac{1}{2}$
  to get a non-exchangeable copula. In particular, picking a concrete example,
  let $\alpha = \frac{1}{4}$, $\beta = \frac{1}{2}$, $\theta = 20$, $u =
  \frac{1}{3}$ and $v = \frac{2}{3}$ yields an $\eta > 0.2$. Taking $\varepsilon
  = 0.1$ yields the counterexample.

  Since we are ultimately interested in approximating copula densities and not
  copulas themselves, it is possible to work with the copula density of the
  previous copula and show that we obtain $\eta > 0.3$ for the same parameter
  values.

The proof for the $L_1 $ norm follows from that of the $L_\infty$ norm because $c$ is continuous
and we only need to look at a compact subset.
\end{proof}

\begin{proof} [\textbf{ Proposition~\ref{Prop: radial symmetry} }]
Proof of (i). Without loss of generality, consider the stochastic representation
$
    \bm {X}:=  \bm {A} \bm {Y}, \quad \bm{U}  =  F (\bm{X}),
$
  where $\bm Y$ is spherically symmetric,  $\bm R = \bm A'\bm  A$ and $ F$ is the distribution of $\bm X$. Then $\bm {oY}
  \xequal{d} \bm Y$ for any orthogonal matrix $\bm o$ because $\bm Y$ is spherically
  symmetric. It implies that $- \bm X \xequal{d} \bm X$ and thus
$
   \bm{1}_M - \bm U = \bm{1}_M - F (\bm X)  =  F (- \bm X) .
$
  Finally, consider a finite mixture of elliptical copulas $ G =  \pi_1 G_{1, \bm{R}_1}+ \cdots +  \pi_R G_{1, \bm{R}_R}$.
  Each component is radially symmetric (around $\frac{1}{2} \bm{1}_M$) implying
  that $ G$ is again radially symmetric.
  \end{proof}

\begin{proof} [\textbf{Proposition \ref{prop: elliptic copulas and radial symmetry}  }]
We give the proof for the $L_\infty$ norm. The proof for the $L_1$ norm then follows from the continuity of $c$ over a compact subset.
  We already showed that $g = \sum_{r = 1}^R \pi_r g_r$ is radially
  symmetric. The copula density $c$ being non-radially symmetric means that
  there exists at least one $\bm u \in (0, 1)^M$, $c (\bm u) \neq c (\bm{1}_M - \bm u)$.
  As in the Archimedean copula
  case, the lack of approximation occurs because $g$ is incapable of
  separating some points that $c$ is capable of separating. For one of those
  points, define, for $\eta > 0$,
  \begin{align*}
    \eta & :=  | c (\bm u) - c (\bm{1}_M - \bm u) |\\
    & =  | c (\bm u) - g (\bm u) + g (\bm u) - g (\bm{1}_M - \bm u) + g (\bm{1}_M - \bm u) - c (\bm{1}_M - \bm u)
    |\\
    & \leqslant  | c (\bm u) - g (\bm u) | + |g (\bm u) - g (\bm{1}_M - \bm u) | + |g (\bm 1_M - u) -
    c (\bm 1_M - \bm u) |\\
    & =  | c (\bm  u) - g (\bm u) | + |g (\bm {1}_M - \bm u) - c (\bm{1}_M - \bm u) |\\
    & \leqslant  2 \max \{ | c (\bm u) - g (\bm u) | \nobracket, \nobracket |g (\bm{1}_M
    - \bm u) - c (\bm{1}_M - \bm u) | \}\\
    & \leqslant 2 \sup_{\bm u \in (0, 1)^M} | c (\bm u) - g (\bm u) | .
  \end{align*}
  We deduce that choosing $\varepsilon = \eta / 2$ is sufficient to prove the
  existence of the counter-example.
  Consider a Clayton copula in two dimensions with parameter $\theta = 1$
  \[ c (u, v):= \frac{2 uv}{(u + v - uv)^3}. \]
  Picking $u = \frac{1}{10}$ and $v = \frac{1}{2}$ allows us to find $\eta >
  0.4$.
\end{proof}

\section{Drawing from the bivariate copula \eqref{eq: example copula}    \label{app: drawing from example copula}}
First, we note that the conditional copula distribution
\begin{eqnarray*}
  C (u|v ) & = & \frac{\partial}{\partial v} C (u, v)\\
  & = & (1 - \beta) u^{1 - \alpha} v^{- \beta}  [u^{- \theta \alpha} + v^{-
  \theta \beta} - 1]^{- \frac{1}{\theta}}\\
  & + & \beta u^{1 - \alpha} v^{- \beta (1 + \theta)}  [u^{- \theta \alpha} +
  v^{- \theta \beta} - 1]^{- \frac{1}{\theta} - 1}
\end{eqnarray*}
can itself be written as a mixture with $\beta \in (0, 1)$ determining the
mixing probability. We can  sample from this model as follows
\begin{itemize}
  \item Draw $V$ from a uniform distribution. Set $X_1 = F_1^{- 1} (V)$.

  \item Draw another independent uniform $W$. Set $U = C^{- 1} (W|V)$. Set
  $X_2 = F_2^{- 1} (U)$.

  Thus, $(U, V)$ is a draw from the copula model and $(X_1, X_2)$ is a draw
  from $f$.
\end{itemize}

\bibliographystyle{apalike}
\bibliography{refs}
\end{document}